\documentclass{article}

\usepackage{times}
\usepackage{graphicx} 
\usepackage{subfigure}

\usepackage{natbib}

\usepackage{algorithm}
\usepackage{algorithmic}
\usepackage[algo2e,linesnumbered,lined,boxed,vlined]{algorithm2e}

\usepackage{amsfonts}
\usepackage{amssymb}

\usepackage{hyperref}

\usepackage[accepted]{icml2016}

\usepackage{comment}
\usepackage{amsthm}
\usepackage{mathtools}
\usepackage{tabularx}
\usepackage{multirow}

\newtheorem{theorem}{Theorem}
\newtheorem{lemma}{Lemma}

\newtheorem{corollary}{Corollary}

\def\b{\ensuremath\boldsymbol}

\usepackage{lastpage}
\usepackage{fancyhdr}
\usepackage{url}
\icmltitlerunning{}

\begin{document}

\AddToShipoutPictureBG*{%
  \AtPageUpperLeft{%
    \setlength\unitlength{1in}%
    \hspace*{\dimexpr0.5\paperwidth\relax}
    \makebox(0,-0.75)[c]{\normalsize {\color{black} To appear as a part of an upcoming textbook on dimensionality reduction and manifold learning.}}
    }}

\twocolumn[
\icmltitle{Uniform Manifold Approximation and Projection (UMAP) and its Variants: \\Tutorial and Survey}

\icmlauthor{Benyamin Ghojogh}{bghojogh@uwaterloo.ca}
\icmladdress{Department of Electrical and Computer Engineering, 
\\Machine Learning Laboratory, University of Waterloo, Waterloo, ON, Canada}
\icmlauthor{Ali Ghodsi}{ali.ghodsi@uwaterloo.ca}
\icmladdress{Department of Statistics and Actuarial Science \& David R. Cheriton School of Computer Science, 
\\Data Analytics Laboratory, University of Waterloo, Waterloo, ON, Canada}
\icmlauthor{Fakhri Karray}{karray@uwaterloo.ca}
\icmladdress{Department of Electrical and Computer Engineering, 
\\Centre for Pattern Analysis and Machine Intelligence, University of Waterloo, Waterloo, ON, Canada}
\icmlauthor{Mark Crowley}{mcrowley@uwaterloo.ca}
\icmladdress{Department of Electrical and Computer Engineering, 
\\Machine Learning Laboratory, University of Waterloo, Waterloo, ON, Canada}

\icmlkeywords{Tutorial}

\vskip 0.3in
]

\begin{abstract}
Uniform Manifold Approximation and Projection (UMAP) is one of the state-of-the-art methods for dimensionality reduction and data visualization. This is a tutorial and survey paper on UMAP and its variants. We start with UMAP algorithm where we explain probabilities of neighborhood in the input and embedding spaces, optimization of cost function, training algorithm, derivation of gradients, and supervised and semi-supervised embedding by UMAP. Then, we introduce the theory behind UMAP by algebraic topology and category theory. Then, we introduce UMAP as a neighbor embedding method and compare it with t-SNE and LargeVis algorithms. We discuss negative sampling and repulsive forces in UMAP's cost function. DensMAP is then explained for density-preserving embedding. We then introduce parametric UMAP for embedding by deep learning and progressive UMAP for streaming and out-of-sample data embedding. 
\end{abstract}

\section{Introduction}

Dimensionality reduction and manifold learning can be used for feature extraction and data visualization. Dimensionality reduction methods can be divided into three categories which are spectral methods, probabilistic methods, and neural network-based methods \cite{ghojogh2021data}. Some of the probabilistic methods are neighbor embedding methods where the probabilities of neighborhoods are used in which attractive and repulsive forces are utilized for neighbor and non-neighbor points, respectively. Some of the well-known neighbor embedding methods are Student's t-distributed Stochastic Neighbor Embedding (t-SNE) \cite{maaten2008visualizing,ghojogh2020stochastic}, LargeVis \cite{tang2016visualizing}, and Uniform Manifold Approximation and Projection (UMAP) \cite{mcinnes2018umap}. Interestingly, both t-SNE and UMAP are state-of-the-art methods for data visualization. 
The reason behind the name of UMAP is that it assumes and approximates that the data points are uniformly distributed on an underlying manifold. The term ``projection" in the name of algorithm is because it sort of projects, or embeds, data onto a subspace for dimensionality reduction. 

The theory behind UMAP is based on algebraic topology and category theory. 
The main idea of UMAP is constructing fuzzy topological representations for both high-dimensional data and low-dimensional embedding of data and changing the embedding so that its fuzzy topological representation becomes similar to that of the high-dimensional data. UMAP has been widely used for DNA and single-cell data visulization and feature extraction \cite{becht2019dimensionality,dorrity2020dimensionality}. It is noteworthy that t-SNE has also been used for single-cell applications \cite{kobak2019art}.  
Some other applications of UMAP are visualizing deep features \cite{carter2019activation}, art \cite{vermeulen2021application}, and visualizing BERT features in natural language processing \cite{coenen2019visualizing,levine2019sensebert}. 
This paper is a tutorial and survey paper on UMAP and its variants. 

The remainder of this paper is organized as follows. We explain the details of UMAP algorithm in Section \ref{section_UMAP} and explain the category theory and algebraic topology behind it in Section \ref{section_categoryTheory_algebraicTopology}. Explaining UMAP as neighbor embedding and comparison with t-SNE and LargeVis are explained in Section \ref{section_compare_tSNE_LargeVis}. Then, we discuss the repulsive forces, negative sampling, and effective cost function of UMAP in Section \ref{section_discussion_repulsive_forces}. DensMAP, parametric UMAP, and progressive UMAP are introduced in Sections \ref{section_DensMAP}, \ref{section_parametric_UMAP}, and \ref{section_progressive_UMAP}, respectively. Finally, Section \ref{section_conclusion} concludes the paper. 

\section*{Required Background for the Reader}

This paper assumes that the reader has general knowledge of calculus, probability, linear algebra, and basics of optimization. 
The required background on algebraic topology and category theory are explained in the paper. 

\section{UMAP}\label{section_UMAP}

\subsection{Data Graph in the Input Space}

Consider a training dataset $\b{X} = [\b{x}_1, \dots, \b{x}_n] \in \mathbb{R}^{d \times n}$ where $n$ is the sample size and $d$ is the dimensionality. We construct a $k$-Nearest Neighbors ($k$NN) graph for this dataset. 
It has been empirically observed that UMAP requires fewer number of neighbors than t-SNE \cite{sainburg2020parametric}. Its default value is $k=15$.
We denote the $j$-th neighbor of $\b{x}_i$ by $\b{x}_{i,j}$.
Let $\mathcal{N}_i$ denote the set of neighbor points for the point $\b{x}_i$, i.e., $\mathcal{N}_i := \{\b{x}_{i,1}, \dots, \b{x}_{i,k}\}$.
We treat neighbor relationship between points stochastically. 
Inspired by SNE \cite{hinton2003stochastic} and t-SNE \cite{maaten2008visualizing,ghojogh2020stochastic}, we use the Gaussian or Radial Basis Function (RBF) kernel for the measure of similarity between points in the input space.
The probability that a point $\b{x}_i$ has the point $\b{x}_j$ as its neighbor can be computed by the similarity of these points:
\begin{align}\label{equation_UMAP_p_Directional}
p_{j|i} := 
\left\{
    \begin{array}{ll}
        \exp\big(\!-\frac{\|\b{x}_i - \b{x}_j\|_2 - \rho_i}{\sigma_i}\big) & \mbox{if } \b{x}_j \in \mathcal{N}_i \\
        0 & \mbox{Otherwise},
    \end{array}
\right.
\end{align}
where $\|.\|_2$ denotes the $\ell_2$ norm.
The $\rho_i$ is the distance from $\b{x}_i$ to its nearest neighbor:
\begin{align}\label{equation_umap_rho}
\rho_i := \min\{\|\b{x}_i - \b{x}_{i,j}\|_2\, |\, 1\leq j\leq k\}.
\end{align}
The $\sigma_i$ is the scale parameter which is calculated such that the total similarity of point $\b{x}_i$ to its $k$ nearest neighbors is normalized. By binary search, we find $\sigma_i$ to satisfy:
\begin{align}\label{equation_umap_sigma}
\sum_{j=1}^k \exp\!\Big(\!\!-\!\frac{\|\b{x}_i - \b{x}_{i,j}\|_2 - \rho_i}{\sigma_i}\Big) = \log_2(k).
\end{align}
Note that t-SNE \cite{maaten2008visualizing} has a similar search for its scale using entropy as perplexity. These searches make the neighborhoods of various points behave similarly because the scale for a point in a dense region of dataset becomes small while the scale of a point in a sparse region of data becomes large. In other words, UMAP and t-SNE both assume (or approximate) that points are uniformly distributed on an underlying low-dimensional manifold. This approximation is also included in the name of UMAP. 

Eq. (\ref{equation_UMAP_p_Directional}) is a directional similarity measure. To have a symmetric measure with respect to $i$ and $j$, we symmetrize it as:
\begin{align}\label{equation_UMAP_p}
\mathbb{R} \ni p_{ij} := p_{j|i} + p_{i|j} - p_{j|i}\, p_{i|j}. 
\end{align}
This is a symmetric measure of similarity between points $\b{x}_i$ and $\b{x}_j$ in the input space. 

\subsection{Data Graph in the Embedding Space}

Let the embeddings of points be $\b{Y} = [\b{y}_1, \dots, \b{y}_n] \in \mathbb{R}^{p \times n}$ where $p$ is the dimensionality of embedding space and is smaller than input dimensionality, i.e., $p \ll d$. 
Note that $\b{y}_i$ is the embedding corresponding to $\b{x}_i$.
In the embedding space, the probability that a point $\b{y}_i$ has the point $\b{y}_j$ as its neighbor can be computed by the similarity of these points:
\begin{align}\label{equation_UMAP_q}
\mathbb{R} \ni q_{ij} := (1 + a\, \|\b{y}_i - \b{y}_j\|_2^{2b})^{-1},
\end{align}
which is symmetric with respect to $i$ and $j$. The variables $a>0$ and $b>0$ are hyperparameters determined by the user. By default, we have $a \approx 1.929$ and $b \approx 0.7915$ \cite{mcinnes2018umap}, although it has been empirically seen that setting $a=b=1$ does not qualitatively impact the results \cite{bohm2020unifying}. 

\subsection{Optimization Cost Function}

UMAP aims to make the data graph in the low-dimensional embedding space similar to the data graph in the high-dimensional embedding space. 
In other words, we treat Eqs. (\ref{equation_UMAP_p}) and (\ref{equation_UMAP_q}) as probability distributions and minimize the difference of these distributions to make similarities of points in the embedding space similar to similarities of points in the input space. 
A measure for the difference of these similarities of graphs is the fuzzy cross-entropy defined as:
\begin{align}\label{equation_umap_cost}
c_1 := \sum_{i=1}^n \sum_{j=1, j \neq i}^n \Big(p_{ij} \ln(\frac{p_{ij}}{q_{ij}}) + (1 - p_{ij}) \ln(\frac{1 - p_{ij}}{1 - q_{ij}})\Big),
\end{align}
where $\ln(.)$ is the natural logarithm. 
The definition of this cross-entropy is in the field of fuzzy category theory which we will explain in Section \ref{section_categoryTheory_algebraicTopology} (see Eq. (\ref{equation_fuzzy_cross_entropy})).

The first term in Eq. (\ref{equation_umap_cost}) is the \textit{attractive force} which attracts the embeddings of neighbor points toward each other. 
This term should only appear when $p_{ij} \neq 0$ which means either $\b{x}_j$ is a neighbor of $\b{x}_i$, or $\b{x}_i$ is a neighbor of $\b{x}_j$, or both (see Eq. (\ref{equation_UMAP_p})). 
The second term in Eq. (\ref{equation_umap_cost}) is the \textit{repulsive force} which repulses the embeddings of non-neighbor points away from each other. As the number of all permutations of non-neighbor points is very large, computation of the second term is non-tractable in big data. Inspired by Word2Vec \cite{mikolov2013distributed} and LargeVis \cite{tang2016visualizing}, UMAP uses \textit{negative sampling} where, for every point $\b{x}_i$, $m$ points are sampled randomly from the training dataset and treat them as non-negative (negative) points for $\b{x}_i$. As the dataset is usually large, i.e. $m \ll n$, the sampled points will be actual negative points with high probability.
The summation over the second term in Eq. (\ref{equation_umap_cost}) is computed only over these negative samples rather than \textit{all} negative points. 


UMAP changes the data graph in the embedding space to make it similar to the data graph in the input space. 
Eq. (\ref{equation_umap_cost}) is the cost function which is minimized in UMAP where the optimization variables are $\{y_i\}_{i=1}^n$:
\begin{align*}
&\min_{\{\b{y}_i\}_{i=1}^n} c_1 := \min_{\{\b{y}_i\}_{i=1}^n} \sum_{i=1}^n \sum_{j=1, j \neq i}^n \Big(p_{ij} \ln(p_{ij}) - p_{ij} \ln(q_{ij}) \\
&~~~~~~~ + (1 - p_{ij}) \ln(1 - p_{ij}) - (1 - p_{ij}) \ln(1 - q_{ij})\Big) \\
&= \!\!\min_{\{\b{y}_i\}_{i=1}^n} -\sum_{i=1}^n \sum_{j=1, j \neq i}^n\!\!\! \Big( p_{ij} \ln(q_{ij}) + (1 - p_{ij}) \ln(1 - q_{ij})\Big) 
\end{align*}
According to Eqs. (\ref{equation_UMAP_p_Directional}), (\ref{equation_UMAP_p}), and (\ref{equation_UMAP_q}), in contrast to $q_{ij}$, the $p_{ij}$ is independent of the optimization variables $\{y_i\}_{i=1}^n$. Hence, we can drop the constant terms to revise the cost function:
\begin{align}\label{equation_umap_cost2}
&c_2 := -\sum_{i=1}^n \sum_{j=1, j \neq i}^n \Big(p_{ij} \ln(q_{ij}) + (1 - p_{ij}) \ln(1 - q_{ij})\Big),
\end{align}
which should be minimized. 
Two important terms in this cost function are:
\begin{align}
&c^a_{i,j} := - \ln(q_{ij}), \label{equation_umap_cost2_attractive} \\
&c^r_{i,j} := - \ln(1-q_{ij}), \label{equation_umap_cost2_repulsive}
\end{align}
and we can write:
\begin{align}
c_2 &:= \sum_{i=1}^n \sum_{j=1, j \neq i}^n \Big(p_{ij}\, c^a_{i,j} + (1 - p_{ij})\, c^r_{i,j}\Big) \label{equation_umap_cost2_withAttractiveAndRepulsiveForces}\\
&\overset{(a)}{=} 2\sum_{i=1}^n \sum_{j=i+1}^n \Big(p_{ij}\, c^a_{i,j} + (1 - p_{ij})\, c^r_{i,j}\Big), \label{equation_umap_cost2_withAttractiveAndRepulsiveForces_2}
\end{align}
where $(a)$ is because $p_{ij}=p_{ji}$, $c_{i,j}^a=c_{j,i}^a$, and $c_{i,j}^r=c_{j,i}^r$ are symmetric. 

The Eqs. (\ref{equation_umap_cost2_attractive}) and (\ref{equation_umap_cost2_repulsive}) are the attractive and repulsive forces in Eq. (\ref{equation_umap_cost2}), respectively. The attractive force attracts the neighbor points toward each other in the embedding space while the repulsive force pushes the non-neighbor points (i.e., points with low probability of being neighbors) away from each other in the embedding space. 
According to Eq. (\ref{equation_umap_cost2_withAttractiveAndRepulsiveForces}), $c^a_{i,j}$ and $c^r_{i,j}$ occur with probability $p_{ij}$ and $(1- p_{ij})$, respectively.
For every point, we call it the \textit{anchor} point and we call its neighbor and non-neighbor points, with large and small $p_{ij}$, as the \textit{positive} and \textit{negative} points, respectively. 

\SetAlCapSkip{0.5em}
\IncMargin{0.8em}
\begin{algorithm2e}[!t]
\DontPrintSemicolon
    \textbf{Input}: Training data $\{\b{x}_i\}_{i=1}^n$\;
    Construct $k$NN graph\;
    Initialize $\{\b{y}_i\}_{i=1}^n$ by Laplacian eigenmap\;
    Calculate $p_{ij}$ and $q_{ij}$ for $\forall i,j \in \{1, \dots, n\}$ by Eqs. (\ref{equation_UMAP_p}) and (\ref{equation_UMAP_q})\;
    $\eta \gets 1$, $\nu \gets 0$\;
    \While{not converged}{
        $\nu \gets \nu + 1$ \quad // epoch index\; 
        \For{$i$ from $1$ to $n$}{
            \For{$j$ from $1$ to $n$}{
                $u \sim U(0,1)$\;
                \If{$u \leq p_{ij}$}{
                    $\b{y}_i \gets \b{y}_i - \eta \frac{\partial c^a_{i,j}}{\partial \b{y}_i}$\;
                    $\b{y}_j \gets \b{y}_j - \eta \frac{\partial c^a_{i,j}}{\partial \b{y}_j}$\;
                    \For{$m$ iterations}{
                        $l \sim U\{1, \dots, n\}$\;
                        $\b{y}_i \gets \b{y}_i - \eta \frac{\partial c^r_{i,l}}{\partial \b{y}_i}$\;
                        // The next line does not exist in original UMAP:\;
                        $\b{y}_l \gets \b{y}_l - \eta \frac{\partial c^r_{i,l}}{\partial \b{y}_l}$ \label{algorithm_umap_update_negativeSample}\;
                    }
                }
            }
        }
        $\eta \gets 1 - \frac{\nu}{\nu_\text{max}}$\;
    }
    \textbf{Return} $\{\b{y}_i\}_{i=1}^n$\;
\caption{UMAP algorithm}\label{algorithm_umap}
\end{algorithm2e}
\DecMargin{0.8em}

\subsection{The Training Algorithm of UMAP}

The procedure of optimization in UMAP is shown in Algorithm \ref{algorithm_umap}. As this algorithm shows, a $k$NN graph is constructed from the training data $\{\b{x}_i\}_{i=1}^n$. UMAP uses Laplacian eigenmap \cite{belkin2001laplacian,ghojogh2021laplacian}, also called spectral embedding, for initializing the embeddings of points denoted by $\{\b{y}_i\}_{i=1}^n$. 
Using Eqs. (\ref{equation_UMAP_p}) and (\ref{equation_UMAP_q}), $p_{ij}$ and $q_{ij}$ are calculated for all points. Stochastic Gradient Descent (SGD) is used for optimization where optimization is performed iteratively. In every iteration (epoch), we iterate over points twice with indices $i$ and $j$ where the $i$-th point is called the \textit{anchor}. For every pair of points $\b{x}_i$ and $\b{x}_j$, we update their embeddings $\b{x}_i$ and $\b{x}_j$ with probability $p_{ij}$ (recall Eq. (\ref{equation_umap_cost2})). 
If $p_{ij}$ is large, it means that the points $\b{x}_i$ and $\b{x}_j$ are probably neighbors (in this case, the $j$-th point is called the \textit{positive} point) and their embeddings are highly likely to be updated to become close in the embedding space based on the attractive force. 
For implementing it, we can sample a uniform value from the continuous uniform distribution $U(0,1)$ and if that is less than $p_{ij}$, we update the embeddings. 
We update the embeddings $\b{y}_i$ and $\b{y}_j$ by gradients $\partial c^a_{i,j} / \partial \b{y}_i$ and $\partial c^a_{i,j} / \partial \b{y}_j$, respectively, where $\eta$ is the learning rate.

For repulsive forces, we use negative sampling as was explained before. If $m$ denotes the size of negative sample, we sample $m$ indices from the discrete uniform distribution $U\{1, \dots, n\}$. These are the indices of points which are considered as \textit{negative} samples $\{\b{y}_l\}$ where $|\{\b{y}_l\}|=m$. 
As the size of dataset is usually large enough to satisfy $n \gg m$, these negative points are probably valid because many of the points are non-neighbors of the considered anchor. 
In negative sampling, the original UMAP \cite{mcinnes2018umap} updates only the embedding of anchor $\b{y}_i$ by gradient of the repulsive force $\partial c^a_{i,j} / \partial \b{y}_i$. One can additionally update the embedding of negative point $\b{y}_l$ by gradient of the repulsive force $\partial c^a_{i,j} / \partial \b{y}_l$ \cite{damrich2021umap}. 
The mentioned gradients are computed in the following lemmas.

\begin{lemma}[\cite{mcinnes2018umap}]
The gradients of attractive and repulsive cost functions in UMAP are:
\begin{align}
& \frac{\partial c^a_{i,j}}{\partial \b{y}_i} = \frac{2ab \|\b{y}_i - \b{y}_j\|_2^{2(b-1)}}{(1 + a\, \|\b{y}_i - \b{y}_j\|_2^{2b})} (\b{y}_i - \b{y}_j), \\
& \frac{\partial c^r_{i,j}}{\partial \b{y}_i} = \frac{-2b}{(\varepsilon+\|\b{y}_i - \b{y}_j\|_2^2) (1 + a\, \|\b{y}_i - \b{y}_j\|_2^{2b})} (\b{y}_i - \b{y}_j), 
\end{align}
where $\varepsilon$ is a small positive number, e.g. $\varepsilon = 0.001$, for stability to prevent division by zero when $\b{y}_i \approx \b{y}_j$.
Likewise, we have:
\begin{align*}
& \frac{\partial c^a_{i,j}}{\partial \b{y}_j} = \frac{2ab \|\b{y}_i - \b{y}_j\|_2^{2(b-1)}}{(1 + a\, \|\b{y}_i - \b{y}_j\|_2^{2b})} (\b{y}_j - \b{y}_i), \\
& \frac{\partial c^r_{i,j}}{\partial \b{y}_j} = \frac{-2b}{(\varepsilon+\|\b{y}_i - \b{y}_j\|_2^2) (1 + a\, \|\b{y}_i - \b{y}_j\|_2^{2b})} (\b{y}_j - \b{y}_i).
\end{align*}
\end{lemma}
\begin{proof}
For the first equation, we have:
\begin{align}
\frac{\partial c^a_{i,j}}{\partial \b{y}_i} &= \frac{\partial c^a_{i,j}}{\partial q_{ij}} \times \frac{\partial q_{ij}}{\partial \b{y}_i} = \frac{-1}{q_{ij}} \times \Big( \frac{-1}{(1 + a\, \|\b{y}_i - \b{y}_j\|_2^{2b})^{2}} \nonumber\\
&\times 2ab (\b{y}_i - \b{y}_j) \times \|\b{y}_i - \b{y}_j\|_2^{2(b-1)} \Big) \nonumber\\
&\overset{(\ref{equation_UMAP_q})}{=} \frac{2ab \|\b{y}_i - \b{y}_j\|_2^{2(b-1)}}{(1 + a\, \|\b{y}_i - \b{y}_j\|_2^{2b})} (\b{y}_i - \b{y}_j).
\end{align}
For the second equation, we have:
\begin{align}
&\frac{\partial c^r_{i,j}}{\partial \b{y}_i} =  \nonumber\\
&\frac{\partial c^r_{i,j}}{\partial q_{ij}} \times \frac{\partial q_{ij}}{\partial \b{y}_i} = \frac{1}{1-q_{ij}} \times \Big( \frac{-1}{(1 + a\, \|\b{y}_i - \b{y}_j\|_2^{2b})^{2}} \nonumber\\
&\times 2ab (\b{y}_i - \b{y}_j) \times \|\b{y}_i - \b{y}_j\|_2^{2(b-1)} \Big) \nonumber\\
&= \frac{-2ab \|\b{y}_i - \b{y}_j\|_2^{2(b-1)}}{(1-q_{ij}) (1 + a\, \|\b{y}_i - \b{y}_j\|_2^{2b})^2} (\b{y}_i - \b{y}_j). \label{equation_proof_derivative_c_r_q}
\end{align}
The term in the numerator can be simplified as:
\begin{align*}
& -2ab \|\b{y}_i - \b{y}_j\|_2^{2(b-1)} \!= -2b (a \|\b{y}_i - \b{y}_j\|_2^{2b}) \|\b{y}_i - \b{y}_j\|_2^{-2} \\
&\overset{(\ref{equation_UMAP_q})}{=} -2b\, (q_{ij}^{-1} - 1) \|\b{y}_i - \b{y}_j\|_2^{-2}. 
\end{align*}
The term in the denominator can be simplified as:
\begin{align*}
& (1-q_{ij}) (1 + a\, \|\b{y}_i - \b{y}_j\|_2^{2b})^2 \overset{(\ref{equation_UMAP_q})}{=} (1-q_{ij}) q_{ij}^{-2} \\
& = q_{ij}^{-2} - q_{ij}^{-1} = q_{ij}^{-1} (q_{ij}^{-1} - 1).
\end{align*}
Hence, Eq. (\ref{equation_proof_derivative_c_r_q}) can be simplified as:
\begin{align*}
\frac{\partial c^r_{i,j}}{\partial \b{y}_i} &= \frac{-2b\, (q_{ij}^{-1} - 1) \|\b{y}_i - \b{y}_j\|_2^{-2}}{q_{ij}^{-1} (q_{ij}^{-1} - 1)} (\b{y}_i - \b{y}_j) \\
&= \frac{-2b}{\|\b{y}_i - \b{y}_j\|_2^2\, q_{ij}^{-1}} (\b{y}_i - \b{y}_j) \\
&\overset{(\ref{equation_UMAP_q})}{=} \frac{-2b}{\|\b{y}_i - \b{y}_j\|_2^2\, (1 + a\, \|\b{y}_i - \b{y}_j\|_2^{2b})} (\b{y}_i - \b{y}_j).
\end{align*}
If we add $\varepsilon$ for stability to the squared distance in the denominator, the equation is obtained. Q.E.D.
\end{proof}


\subsection{Supervised and Semi-supervised Embedding}\label{section_supervised_semisupervised_umap}

The explained UMAP algorithm is unsupervised. 
We can have supervised and semi-supervised embedding by UMAP \cite{sainburg2020parametric}. For supervised UMAP, we can use UMAP cost function, Eq. (\ref{equation_umap_cost2}), regularized by a classification cost function such as cross-entropy or triplet loss. 
In semi-supervised case, some part of dataset has labels and some part does not. We can iteratively alternate between UMAP's cost function, Eq. (\ref{equation_umap_cost2}), and a classification cost function. In this way, the embeddings are updated by UMAP and fine-tuned by class labels and this procedure is repeated iteratively until convergence of embedding. 


\section{Justifying UMAP's Cost Function by Algebraic Topology and Category Theory}\label{section_categoryTheory_algebraicTopology}

UMAP is a neighbor embedding method where the probability of neighbors for every point is used for optimization of embedding. However, its cost function, Eq. (\ref{equation_umap_cost}), can be justified by algebraic topology \cite{may1992simplicial,friedman2012survey} and category theory \cite{mac2013categories,riehl2017category}. Specifically, its theory heavily uses the fuzzy category theory \cite{spivak2012metric}. 
In this section, we briefly introduce the theory behind UMAP. The reader can skip this section if they do not wish to know the theory behind UMAP's cost function.

First, we introduce some concepts which will be used for theory of UMAP:
\begin{itemize}
\item A \textit{simplex} is a generalization of triangle to arbitrary dimensions. 
\item A \textit{simplicial complex} is a set of points, line segments, triangles, and their $d$-dimensional counterparts \cite{may1992simplicial}. 
\item A \textit{fuzzy set} is a mapping $\mu: \mathcal{A} \rightarrow [0,1]$ from carrier set $\mathcal{A}$ where the mapping is called the membership function \cite{zadeh1965fuzzy}. We can denote a fuzzy set by $(\mathcal{A}, \mu)$.

\item In category theory, a \textit{category} is a collection of objects which are linked by arrows. For example, objects can be sets and the arrows can be functions between sets \cite{mac2013categories}. 
A \textit{morphism} is a mapping from a mathematical structure to another structure without changing type (e.g., morphisms are functions in set theory). 
A \textit{functor} is defined as a mapping between categories. 
\textit{Adjunction} is a relation between two functors. The two functors having an adjunction are \textit{adjoint} functors, one of which is the \textit{left adjoint} and the other is the \textit{right adjoint}. 

\item A \textit{topology} is a geometrical object which is still preserved by continuous deformations such as stretching and twisting but without tearing and making or closing holes. A \textit{topological space} is a set of topologies whose operations are continuous deformations. 

\item We define the category $\b{\Delta}$ whose objects are finite order sets $[n] = \{1, \dots, n\}$ with order-preserving maps as its morphisms {\citep[Definition 1]{mcinnes2018umap}}. 
A simplicial set is a functor from $\b{\Delta}$ to the category of sets {\citep[Definition 2]{mcinnes2018umap}}. 

\item Consider a category of fuzzy sets, denoted by \textbf{Fuzz} {\citep[Definition 4]{mcinnes2018umap}}. The category of fuzzy simplicial sets, denoted by \textbf{sFuzz}, is the category of objects with functors from $\b{\Delta}$ to \textbf{Fuzz} and natural transformations as its morphisms {\citep[Definition 5]{mcinnes2018umap}}. 

\item An extended-pseudo-metric space is a set $\mathcal{X}$ and a mapping $d: \mathcal{X} \times \mathcal{X} \rightarrow \mathbb{R}_{\geq 0} \cup \{\infty\}$ where for $x,y \in \mathcal{X}$, we have $d(x,y) \geq 0$ and $x=y \implies d(x,y)=0$ and $d(x,y) = d(y,x)$ and $d(x,z) \leq d(x,y) + d(y,z)$ {\citep[Definition 6]{mcinnes2018umap}}. The mapping $d$ can be seen as a distance metric or pseudo-metric.

\item Let \textbf{Fin-sFuzz} be the sub-category of bounded fuzzy simplicial sets. Let \textbf{FinEPMet} be the sub-category of finite extended-pseudo-metric spaces. 

\end{itemize}

\begin{theorem}[{\citep[Theorem 1]{mcinnes2018umap}}]
The functors \textbf{FinReal}: \textbf{Fin-sFuzz} $\rightarrow$ \textbf{FinEPMet} and \textbf{FinSing}: \textbf{FinEPMet} $\rightarrow$ \textbf{Fin-sFuzz} form an adjunction with \textbf{FinReal} and \textbf{FinSing} as the left and right adjoints. 
\end{theorem}
\begin{proof}
Proof is available in {\citep[Appendix B]{mcinnes2018umap}}. 
\end{proof}

The above theorem shows that we can convert an extended-pseudo-metric space to a fuzzy simplicial set and vice versa. In other words, we have a fuzzy simplicial representation of data space. 
Hence, we have the following corollary. 

\begin{corollary}[Fuzzy Topological Representation {\citep[Definition 9]{mcinnes2018umap}}]
Consider a dataset $\mathcal{X} := \{\b{x}_i \in \mathbb{R}^d\}_{i=1}^n$ lying on an underlying manifold $\mathcal{M}$. Let $\{(\mathcal{X}, d_i)\}_{i=1}^n$ be a family of extended-pseudo-metric spaces with common carrier set $\mathcal{X}$ such that:
\begin{align}
d_i(\b{x}_j, \b{x}_l) := 
\left\{
    \begin{array}{ll}
        d_\mathcal{M}(\b{x}_j, \b{x}_l) - \rho_i & \mbox{if } i=j \text{ or } i=l, \\
        \infty & \mbox{Otherwise},
    \end{array}
\right.
\end{align}
where $d_\mathcal{M}(.,.)$ is the geodesic (shortest) distance on manifold and $\rho_i$ is the distance to the nearest neighbor of $\b{x}_i$ (see Eq. (\ref{equation_umap_rho})). 
The \textit{fuzzy topological representation} of dataset $\mathcal{X}$ is:
\begin{align*}
\bigcup_{i=1}^n \textbf{FinSing}((\mathcal{X}, d_i)),
\end{align*}
where $\bigcup$ is the fuzzy set union. 
\end{corollary}

UMAP creates a fuzzy topological representation for the high-dimensional dataset. Then, it initializes a low-dimensional embedding of dataset and creates a fuzzy topological representation for the low-dimensional embedding of dataset. Then, it tries to modify the low-dimensional embedding of dataset in a way that the fuzzy topological representation of embedding becomes similar to the fuzzy topological representation of high-dimensional data. A measure of difference of two fuzzy topological representations $(\mathcal{A},\mu_1)$ and $(\mathcal{A},\mu_2)$ is their cross-entropy defined as \cite{mcinnes2018umap}:
\begin{align}\label{equation_fuzzy_cross_entropy}
&c\big((\mathcal{A},\mu_1), (\mathcal{A},\mu_2)\big) := \nonumber\\
&\sum_{a \in \mathcal{A}} \Big(\mu_1(a) \ln\Big(\frac{\mu_1(a)}{\mu_2(a)}\Big) + (1-\mu_1(a)) \ln\Big(\frac{1-\mu_1(a)}{1-\mu_2(a)}\Big)\Big).
\end{align}
UMAP minimizes this cross-entropy by changing the embedding iteratively. Hence, it uses cross-entropy as its cost function, i.e., Eq. (\ref{equation_umap_cost}).  

\section{Neighbor Embedding: Comparison with t-SNE and LargeVis}\label{section_compare_tSNE_LargeVis}

UMAP has a connection with t-SNE \cite{maaten2008visualizing,ghojogh2020stochastic} and LargeVis \cite{tang2016visualizing}.
This connection is explained in {\citep[Appendix C]{mcinnes2018umap}}. 
In fact, all UMAP, t-SNE, and LargeVis are neighbor embedding methods in which attractive and repulsive forces are used \cite{bohm2020unifying}. As was explained before, for every point considered as anchor, attractive forces are used for pushing neighbor (also called positive) points to anchor and repulsive forces are used for pulling non-neighbor (also called negative) points away from the anchor points. Note that the ideas of triplet and contrastive losses as well as Fisher discriminant analysis are the same \cite{ghojogh2020fisher}. 
An empirical comparison of UMAP and t-SNE is also available in \cite{repke2021robust}. 

\textbf{-- Comparison of probabilities:}
In t-SNE, the probabilities in the input and embedding spaces are \cite{ghojogh2020stochastic}:
\begin{align}
&p_{j|i} := \frac{\exp\big(\!-\frac{\|\b{x}_i - \b{x}_j\|_2}{\sigma_i}\big)}{\sum_{k=1, k \neq i}^n \exp\big(\!-\frac{\|\b{x}_i - \b{x}_k\|_2}{\sigma_i}\big)}, \label{equation_tSNE_p_Directional} \\
&p_{ij} := \frac{p_{j|i} + p_{i|j}}{2n}, \label{equation_tSNE_p} \\
&q_{ij} := \frac{(1+\|\b{y}_i - \b{y}_j\|_2^2)^{-1}}{\sum_{k=1}^n \sum_{l=1, l \neq k}^{n} (1+\|\b{y}_k - \b{y}_l\|_2^2)^{-1}}, \label{equation_tSNE_q}
\end{align}
where $p_{i|i} = 0, \forall i$.
The $p_{j|i}$ probabilities can be computed for $k$NN graph where $p_{j|i}$ is set to zero for non-neighbor points in the $k$NN graph.
LargeVis uses the same $p_{ij}$ probabilities as t-SNE but approximates the $k$NN to compute it very fast and become more efficient. In LargeVis, the probability in the embedding space is:
\begin{align}
&q_{ij} := (1+\|\b{y}_i - \b{y}_j\|_2^2)^{-1}. \label{equation_LargeVis_q}
\end{align}

Comparing Eqs. (\ref{equation_UMAP_p_Directional}) and (\ref{equation_tSNE_p_Directional}) shows that UMAP, t-SNE, and LargeVis all use Gaussian or RBF kernel for probabilities in the input space.
Comparing Eqs. (\ref{equation_UMAP_p}) and (\ref{equation_tSNE_p}) shows that UMAP and t-SNE/LargeVis use different approaches for symmetrizing the probabilities in the input space. 
Comparing Eqs. (\ref{equation_UMAP_q}), (\ref{equation_tSNE_q}), and (\ref{equation_LargeVis_q}) shows that, in contrast to t-SNE, UMAP and LargeVis do not normalize the probabilities in the embedding space by all pairs of points. This advantage makes UMAP much faster than t-SNE and also makes it more suitable for mini-batch optimization in deep learning (this will be explained more in Section \ref{section_parametric_UMAP}). 
Comparing Eqs. (\ref{equation_UMAP_q}), (\ref{equation_tSNE_q}), and (\ref{equation_LargeVis_q}) also shows that UMAP, t-SNE, and LargeVis all use Cauchy distribution for probabilities in the embedding space. In fact if we set $a=b=1$ in Eq. (\ref{equation_UMAP_q}), it is exactly the same as Eq. (\ref{equation_tSNE_q}) up to the scale of normalization. 

\textbf{-- Comparison of cost functions:}
The cost function in t-SNE, to be minimized, is the KL-divergence \cite{kullback1951information} of the probabilities in the input and embedding spaces:
\begin{equation}\label{equation_tSNE_cost}
\begin{aligned}
c_4 &:= \sum_{i=1}^n \sum_{j=1, j \neq i}^n p_{ij} \ln(\frac{p_{ij}}{q_{ij}}) \\
&= \sum_{i=1}^n \sum_{j=1, j \neq i}^n \big( p_{ij} \ln(p_{ij}) - p_{ij} \ln(q_{ij}) \big),
\end{aligned}
\end{equation}
where Eqs. (\ref{equation_tSNE_p}) and (\ref{equation_tSNE_q}) are used. 
The cost function of LargeVis, to be minimized, is a negative likelihood function stated below:
\begin{equation}\label{equation_LargeVis_cost}
\begin{aligned}
c_5 &:= -\sum_{i=1}^n \sum_{j=1, j \neq i}^n \big(p_{ij} \ln(q_{ij}) + \lambda \ln(1 - q_{ij}) \big),
\end{aligned}
\end{equation}
where $\lambda$ is the regularization parameter and Eqs. (\ref{equation_tSNE_p}) and (\ref{equation_LargeVis_q}) are used. 
Comparing Eqs. (\ref{equation_umap_cost2}), (\ref{equation_tSNE_cost}), and (\ref{equation_LargeVis_cost}) shows that UMAP, t-SNE, and LargeVis have similar, but not exactly equal, cost functions. 
The first term in all these cost functions is responsible for the attractive forces and the second term is for the repulsive forces; hence, they can all be considered as neighbor embedding methods \cite{bohm2020unifying}. 

\section{Discussion on Repulsive Forces and Negative Sampling in the UMAP's Cost Function}\label{section_discussion_repulsive_forces}

\subsection{UMAP's Emphasis on Repulsive Forces}

The gradient of UMAP's cost function, Eq. (\ref{equation_umap_cost2_withAttractiveAndRepulsiveForces}), in epoch $\nu$ is denoted by $(\partial c_2 / \partial \b{y}_i) |_\nu$ and it can be stated as \cite{damrich2021umap}:
\begin{align}\label{equation_umap_cost_epoch}
\frac{\partial c_2}{\partial \b{y}_i}\bigg|_\nu = \sum_{j=1}^n \Big( \mathbb{I}_{ij}^\nu\, \frac{\partial c_{i,j}^a}{\partial \b{y}_i} + \mathbb{I}_{ji}^\nu\, \frac{\partial c_{j,i}^a}{\partial \b{y}_i} + \mathbb{I}_{ij}^\nu\, \sum_{l=1}^n \mathbb{I}_{ijl}^\nu\, \frac{\partial c_{i,l}^r}{\partial \b{y}_i} \Big),
\end{align}
where $c_{i,j}^a$ and $c_{i,j}^r$ are defined in Eqs. (\ref{equation_umap_cost2_attractive}) and (\ref{equation_umap_cost2_attractive}), respectively, and $\mathbb{I}_{ij}^\nu$ is a binary random variable which is one if the points $\b{y}_i$ and $\b{y}_j$ are randomly selected in epoch $\nu$ and otherwise it is zero. Also, $\mathbb{I}_{ijl}^\nu$ is a binary random variable which is one if the point $\b{y}_l$ is one of the negative samples which is randomly sampled for the pair $\b{y}_i$ and $\b{y}_j$ in epoch $\nu$ and otherwise it is zero. 
Recall that the original UMAP does not update the embedding of negative sample itself so we do not have a term for that update in this gradient. 

Eq. (\ref{equation_umap_cost_epoch}) is only for one epoch. 
The expectation of Eq. (\ref{equation_umap_cost_epoch}) over all epochs is \cite{damrich2021umap}:
\begin{align}
&\mathbb{E}\Big[\frac{\partial c_2}{\partial \b{y}_i}\Big|_\nu\Big] \nonumber\\
&= \mathbb{E}\Big[\sum_{j=1}^n \Big( \mathbb{I}_{ij}^\nu\, \frac{\partial c_{i,j}^a}{\partial \b{y}_i} + \mathbb{I}_{ji}^\nu\, \frac{\partial c_{j,i}^a}{\partial \b{y}_i} + \mathbb{I}_{ij}^\nu\, \sum_{l=1}^n \mathbb{I}_{ijl}^\nu\, \frac{\partial c_{i,l}^r}{\partial \b{y}_i} \Big)\Big] \nonumber\\
&\overset{(a)}{=} \sum_{j=1}^n \Big( \mathbb{E}[\mathbb{I}_{ij}^\nu]\, \frac{\partial c_{i,j}^a}{\partial \b{y}_i} + \mathbb{E}[\mathbb{I}_{ji}^\nu]\, \frac{\partial c_{j,i}^a}{\partial \b{y}_i}\Big) \nonumber\\
&~~~~~~~+ \sum_{j=1}^n \sum_{l=1}^n \mathbb{E}[\mathbb{I}_{ij}^\nu\, \mathbb{I}_{ijl}^\nu]\, \frac{\partial c_{i,l}^r}{\partial \b{y}_i} \nonumber\\
&\overset{(b)}{=} \sum_{j=1}^n \Big( p_{ij} \frac{\partial c_{i,j}^a}{\partial \b{y}_i} + p_{ji} \frac{\partial c_{j,i}^a}{\partial \b{y}_i}\Big) + \sum_{j=1}^n \sum_{l=1}^n p_{ij} \frac{m}{n} \frac{\partial c_{i,l}^r}{\partial \b{y}_i} \nonumber
\end{align}
\begin{align}
&\overset{(c)}{=} \sum_{j=1}^n \Big( p_{ij} \frac{\partial c_{i,j}^a}{\partial \b{y}_i} + p_{ji} \frac{\partial c_{j,i}^a}{\partial \b{y}_i}\Big) + \frac{m}{n} \underbrace{\sum_{j=1}^n p_{ij}}_{=\, d_i} \sum_{l=1}^n \frac{\partial c_{i,l}^r}{\partial \b{y}_i} \nonumber\\
&\overset{(d)}{=} \sum_{j=1}^n \Big( p_{ij} \frac{\partial c_{i,j}^a}{\partial \b{y}_i} + p_{ji} \frac{\partial c_{j,i}^a}{\partial \b{y}_i}\Big) + \frac{d_i m}{n} \sum_{j=1}^n \frac{\partial c_{i,j}^r}{\partial \b{y}_i} \nonumber\\
&\overset{(e)}{=} 2 \sum_{j=1}^n \Big( p_{ij} \frac{\partial c_{i,j}^a}{\partial \b{y}_i} + \frac{d_i m}{2n} \frac{\partial c_{i,j}^r}{\partial \b{y}_i} \Big), \label{equation_umap_gradient_cost_epoch}
\end{align}
where $(a)$ is because expectation is a linear operator, $(b)$ is because $\mathbb{E}[\mathbb{I}_{ij}^\nu] = p_{ij}$ and $\mathbb{E}[\mathbb{I}_{ij}^\nu\, \mathbb{I}_{ijl}^\nu] = \mathbb{E}[\mathbb{I}_{ijl}^\nu\, |\, \mathbb{I}_{ij}^\nu] \times \mathbb{E}[\mathbb{I}_{ij}^\nu] = \frac{m}{n} \times p_{ij}$, $(c)$ is because we define the degree of the $i$-th point (node) in the $k$NN graph as $d_i := \sum_{j=1}^n p_{ij}$, $(d)$ is because we change the dummy variable $l$ to $j$ in the last summation, and $(e)$ is because $p_{ij} = p_{ji}$ and $c_{i,j}^a = c_{j,i}^a$ are symmetric. 

On the other hand, according to Eq. (\ref{equation_umap_cost2_withAttractiveAndRepulsiveForces_2}), the gradient of UMAP's cost function, Eq. (\ref{equation_umap_cost2}), can be stated as \cite{damrich2021umap}:
\begin{align}\label{equation_umap_gradient_withAttractiveRepulsiveDerivatives}
\frac{\partial c_2}{\partial \b{y}_i} = 2\sum_{j=1}^n \Big(p_{ij}\, \frac{\partial c^a_{i,j}}{\partial \b{y}_i} + (1 - p_{ij})\, \frac{\partial c^r_{i,j}}{\partial \b{y}_i}\Big).
\end{align}
Eq. (\ref{equation_umap_gradient_withAttractiveRepulsiveDerivatives}) is the gradient of the original UMAP's loss function while Eq. (\ref{equation_umap_gradient_cost_epoch}) is the expected gradient of UMAP's loss function.
Comparing Eqs. (\ref{equation_umap_gradient_cost_epoch}) and (\ref{equation_umap_gradient_withAttractiveRepulsiveDerivatives}) shows that the original UMAP puts more emphasis on negative samples (or repulsive forces) compared to the expected UMAP's loss function because for a negative sample we have $1-p_{ij} \approx 1$ while $d_i m /n \approx 0$ because $m \ll n$. Therefore, UMAP is mistakenly putting more emphasis on negative sampling (or repulsive forces) than required \cite{damrich2021umap}. This has also been empirically investigated in \cite{bohm2020unifying} that negative sampling (or repulsive forces) in UMAP has more weight than required. 

\subsection{UMAP's Effective Cost Function}

The original UMAP does not update the embedding of negative samples themselves. If we also update them, i.e. we perform line \ref{algorithm_umap_update_negativeSample} in Algorithm \ref{algorithm_umap}, the gradient of UMAP's cost function, Eq. (\ref{equation_umap_cost2_withAttractiveAndRepulsiveForces}), in epoch $\nu$ can be stated as \cite{damrich2021umap}:
\begin{align}\label{equation_umap_cost_epoch_withNegativeSampleUpdate}
\frac{\partial c_2}{\partial \b{y}_i}\bigg|_\nu = &\sum_{j=1}^n \Big( \mathbb{I}_{ij}^\nu\, \frac{\partial c_{i,j}^a}{\partial \b{y}_i} + \mathbb{I}_{ji}^\nu\, \frac{\partial c_{j,i}^a}{\partial \b{y}_i} \nonumber\\
&+ \mathbb{I}_{ij}^\nu\, \sum_{l=1}^n \mathbb{I}_{ijl}^\nu\, \frac{\partial c_{i,l}^r}{\partial \b{y}_i} + \sum_{k=1}^n \mathbb{I}_{jk}^\nu\, \mathbb{I}_{jki}^\nu\, \frac{\partial c_{j,i}^r}{\partial \b{y}_i} \Big).
\end{align}
If we do reverse engineering to find the cost function from its gradient, the cost function at epoch $\nu$ becomes:
\begin{align}
c_2 \big|_\nu := \sum_{i=1}^n \sum_{j=1, j\neq i}^n \Big( \mathbb{I}_{ij}^\nu\, c_{i,j}^a + \sum_{l=1}^n \mathbb{I}_{ij}^\nu\, \mathbb{I}_{ijl}^\nu\, c_{i,l}^r \Big).
\end{align}
This is the cost at one epoch. The expectation of this cost over all epochs is \cite{damrich2021umap}:
\begin{align}
c_2 &= \mathbb{E}[c_2 \big|_\nu] \nonumber\\
&= \sum_{i=1}^n \sum_{j=1, j\neq i}^n \Big( \mathbb{E}[\mathbb{I}_{ij}^\nu]\, c_{i,j}^a + \sum_{l=1}^n \mathbb{E}[\mathbb{I}_{ij}^\nu\, \mathbb{I}_{ijl}^\nu]\, c_{i,l}^r \Big) \nonumber\\
&= \sum_{i=1}^n \sum_{j=1, j\neq i}^n \Big( \mathbb{E}[\mathbb{I}_{ij}^\nu]\, c_{i,j}^a \Big) \nonumber\\
&~~~~~~~~~~~~~~ + \sum_{i=1}^n \sum_{j=1, j\neq i}^n  \sum_{l=1}^n \Big( \mathbb{E}[\mathbb{I}_{ij}^\nu\, \mathbb{I}_{ijl}^\nu]\, c_{i,l}^r \Big) \nonumber\\
&\overset{(a)}{=} \sum_{i=1}^n \sum_{j=1, j\neq i}^n \Big( p_{ij}\, c_{i,j}^a \Big) + \sum_{i=1}^n \sum_{j=1, j\neq i}^n  \sum_{l=1}^n \Big( \frac{m}{n} p_{ij}\, c_{i,l}^r \Big) \nonumber\\
&= \sum_{i=1}^n \sum_{j=1, j\neq i}^n \Big( p_{ij}\, c_{i,j}^a \Big) \nonumber\\
&~~~~~~~~~~~~~~ + \frac{m}{n} \Big(\sum_{i=1}^n\underbrace{\sum_{j=1}^n p_{ij}}_{=d_i} + \sum_{j=1}^n\underbrace{\sum_{i=1}^n p_{ji}}_{=d_j} \Big) \sum_{l=1}^n c_{i,l}^r \nonumber\\
&\overset{(b)}{=} \sum_{i=1}^n \sum_{j=1,j \neq i}^n \Big( p_{ij}\, c_{i,j}^a \Big) + \sum_{i=1}^n \sum_{j=i+1}^n \frac{(d_i+d_j) m}{n} c_{i,j}^r \nonumber\\
&\overset{(c)}{=} 2 \sum_{i=1}^n \sum_{j=i+1}^n \Big( p_{ij}\, c_{i,j}^a + \frac{(d_i+d_j) m}{2n} c_{i,j}^r \Big), \label{equation_umap_effective_loss}
\end{align}
where $(a)$ is because $\mathbb{E}[\mathbb{I}_{ij}^\nu] = p_{ij}$ and $\mathbb{E}[\mathbb{I}_{ij}^\nu\, \mathbb{I}_{ijl}^\nu] = \mathbb{E}[\mathbb{I}_{ijl}^\nu\, |\, \mathbb{I}_{ij}^\nu] \times \mathbb{E}[\mathbb{I}_{ij}^\nu] = \frac{m}{n} \times p_{ij}$, $(b)$ is because we change the dummy variable $l$ to $j$ in the last summation, and $(c)$ is because of symmetry of terms in the fist summations with respect to $i$ and $j$. 
Eq. (\ref{equation_umap_effective_loss}) can be considered as UMAP's effective cost function \cite{damrich2021umap} because it also updates the embedding of negative samples by line \ref{algorithm_umap_update_negativeSample} in Algorithm \ref{algorithm_umap}. 

Comparing UMAP's cost, Eq. (\ref{equation_umap_cost2_withAttractiveAndRepulsiveForces_2}), with UMAP's effective cost, Eq. (\ref{equation_umap_effective_loss}), shows that the weight of negative sample (or repulsive forces) should be $\frac{(d_i+d_j) m}{2n}$ rather than $(1-p_{ij})$ if we also update the embeddings of negative samples. As we have $m \ll n$ and $p_{ij}$ is small for negative samples, this weight is much less than the weight in the original UMAP. 

\section{DensMAP for Density-Preserving Embedding}\label{section_DensMAP}

As was explained before, UMAP uses a binary search for the scale  of each point, $\sigma_i$, to satisfy Eq. (\ref{equation_umap_sigma}), so as t-SNE \cite{maaten2008visualizing} which has a similar search for its scale using entropy as perplexity. The search makes the neighborhoods of various points behave similarly so UMAP and t-SNE both assume that points are uniformly distributed on an underlying low-dimensional manifold. Hence, UMAP ignores the density of data around every point by canceling the effect of density with binary search for scales of points. DensMAP \cite{narayan2020density} regularizes the cost function of UMAP to take into account and add back the information of density around each point. It is shown empirically that this consideration of density information results in better embedding \cite{narayan2020density} although we will have more computation for calculation of the regularization term. 

If the neighbors of a point are very close to it, that region is dense for that point. Therefore, a measure of local density can be the local radius defined as the expected (average) distances from neighbors. We denote the local densities in the input and embedding spaces by:
\begin{align}
& R_p(\b{x}_i) := \mathbb{E}_{j\sim p}\big[\|\b{x}_i - \b{x}_j\|_2^2\big] = \frac{\sum_{j=1}^n p_{ij} \|\b{x}_i - \b{x}_j\|_2^2}{\sum_{j=1}^n \|\b{x}_i - \b{x}_j\|_2^2}, \\
& R_q(\b{y}_i) := \mathbb{E}_{j\sim q}\big[\|\b{y}_i - \b{y}_j\|_2^2\big] = \frac{\sum_{j=1}^n q_{ij} \|\b{y}_i - \b{y}_j\|_2^2}{\sum_{j=1}^n \|\b{y}_i - \b{y}_j\|_2^2}.
\end{align}
As the volume of points is proportional to the powers of radius (e.g., notice that the volume of three dimensional sphere is proportional to radius to the power three), the relation of local densities in the input and embedding spaces can be:
\begin{align}
R_q(\b{y}_i) = \alpha\, \big(R_p(\b{x}_i)\big)^\beta \implies r_q^i = \beta\, r_p^i + \gamma,
\end{align}
where $r_q^i := \ln(R_q(\b{y}_i))$, $r_p^i := \ln(R_p(\b{y}_i))$, and $\gamma := \ln(\alpha)$. Therefore, the relation of logarithms of the local densities should be affine dependence. A measure of linear (or affine) dependence is correlation so we use the correlation of logarithms of local densities:
\begin{align}
\text{Corr}(r_q,r_p) := \frac{\text{Cov}(r_q, r_p)}{\sqrt{\text{Var}(r_q) \text{Var}(r_p)}},
\end{align}
where the covariance and variance of densities are:
\begin{align*}
& \text{Cov}(r_q, r_p) := \frac{1}{n-1} \sum_{i=1}^n \Big[(r_q^i - \mu_q) (r_p^i - \mu_p) \Big] \\
& \text{Var}(r_q) := \frac{1}{n-1} \sum_{i=1}^n (r_q^i - \mu_q)^2, 
\end{align*}
where $\mu_q := (1/n) \sum_{j=1}^n r_q^j$, $\mu_p := (1/n) \sum_{j=1}^n r_p^j$, and $\text{Var}(r_p)$ is defined similarly. 
The cost function of densMAP, to be minimized, is the UMAP's cost function regularized by maximization of the correlation of local densities \cite{narayan2020density}:
\begin{align}
c_6 := c_2 - \lambda\, \text{Corr}(r_q,r_p),
\end{align}
where $\lambda$ is the regularization parameter which weights the correlation compared to the UMAP's original cost. 
The gradient of $c_2$ is the gradient of original UMAP discussed before. The gradient of the correlation term is \cite{narayan2020density}:
\begin{align}
\frac{\partial c_6}{\partial \b{y}_i} = \sum_{i=1}^n \sum_{j=1, j \neq i}^n \frac{\partial \text{Corr}(r_q,r_p)}{\partial d_{ij}^2} (y_i - y_j), 
\end{align}
where $d_{ij}^2 := \|\b{y}_i - \b{y}_j\|_2^2$ and:
\begin{align*}
& \frac{\partial \text{Corr}(r_q,r_p)}{\partial d_{ij}^2} = \frac{1}{(n-1) \text{Var}(r_q)^{(3/2)}} \times \\
&~~~~~~~ \bigg[\text{Var}(r_q)\big(r_p^i \frac{\partial r_q^i}{\partial d_{ij}^2} + r_p^j \frac{\partial r_q^j}{\partial d_{ij}^2}\big) \\
&~~~~~~~ - \text{Cov}(r_q, r_p)\big((r_q^i - \mu_q) \frac{\partial r_q^i}{\partial d_{ij}^2} + (r_p^j - \mu_q) \frac{\partial r_q^j}{\partial d_{ij}^2}\big)\bigg],
\end{align*}
and:
\begin{align*}
& \frac{\partial r_q^i}{\partial d_{ij}^2} = (1 + a d_{ij}^{2b})^{-2} \Big[ab d_{ij}^{2(b-1)} + e^{-r_q^i} (1 + a(1-b) d_{ij}^2)\Big].
\end{align*}
Proofs of these derivatives are available in {\citep[Supplementary Note 2]{narayan2020density}}. As in UMAP, DensMAP uses stochastic gradient descent for optimization. Except for the cost function, the algorithm of DensMAP is the same as UMAP. 

\SetAlCapSkip{0.5em}
\IncMargin{0.8em}
\begin{algorithm2e}[!t]
\DontPrintSemicolon
    \textbf{Input}: Training data $\{\b{x}_i\}_{i=1}^n$, learning rate $\eta$\;
    Construct $k$NN graph\;
    Initialize weights $\theta$ of network randomly\;
    Initialize $\{\b{y}_i\}_{i=1}^n \gets \{f_{\theta}(\b{x}_i)\}_{i=1}^n$\;
    Calculate $p_{ij}$ and $q_{ij}$ for $\forall i,j \in \{1, \dots, n\}$ by Eqs. (\ref{equation_UMAP_p}) and (\ref{equation_UMAP_q})\;
    // make batches:\;
    \For{$s$ from $1$ to $\lfloor n/b \rfloor$}{
        $\mathcal{B}_s \gets \{\b{x}_{(s-1)b+1}, \dots, \b{x}_{sb}\}$\;
        // We denote $\mathcal{B}_s = \{\b{x}_1^{(s)}, \dots, \b{x}_b^{(s)}\}$\;
    }
    $\nu \gets 0$\;
    \While{not converged}{
        $\nu \gets \nu + 1$ \quad // epoch index\; 
        // optimize over every mini-batch:\;
        \For{$s$ from $1$ to $\lfloor n/b \rfloor$}{
            $\{\b{y}_i^{(s)}\}_{i=1}^b \gets \{f_{\theta}(\b{x}_i^{(s)})\}_{i=1}^b$\;
            $c^a_s \gets 0, c^r_s \gets 0$\;
            \For{$i$ from $1$ to $b$}{
                \For{$j$ from $1$ to $b$}{
                    $u \sim U(0,1)$\;
                    \If{$u \leq p_{ij}$}{
                        $q_{ij} = (1 + a \|\b{y}_i^{(s)} - \b{y}_j^{(s)}\|_2^{2b})^{-1}$\;
                        $c^a_s = c^a_s +(- \ln(q_{ij}))$\;
                        \For{$m$ iterations}{
                            $l \sim U\{1, \dots, b\}$\;
                            $q_{il} = (1 + a \|\b{y}_i^{(s)} - \b{y}_l^{(s)}\|_2^{2b})^{-1}$\;
                            $c^r_s = c^r_s +(- \ln(1-q_{il}))$\;
                        }
                    }
                }
            }
            $\theta \gets$ backpropagate with loss $(c^a_s + c^r_s)$\;
        }
    }
    \textbf{Return} $\{\b{y}_i\}_{i=1}^n \gets \{f_{\theta}(\b{x}_i)\}_{i=1}^n$\;
\caption{Parametric UMAP algorithm}\label{algorithm_parametric_umap}
\end{algorithm2e}
\DecMargin{0.8em}

\section{Parametric UMAP for Embedding by Deep Learning}\label{section_parametric_UMAP}

Dimensionality reduction algorithms can have their parametric version in which the cost function of the algorithm is used as the loss function of a neural network and the parameters (i.e., weights) of network are trained by backpropagating the error of loss function. 
Inspired by parametric t-SNE \cite{van2009learning}, we can have parametric UMAP \cite{sainburg2020parametric}. Parametric UMAP uses UMAP's cost function, Eq. (\ref{equation_umap_cost}), as the loss function of a neural network for deep learning. It optimizes this cost in mini-batches rather than on the whole dataset; therefore, it can be used for embedding large datasets. 
Also, because of nonlinearity of neural networks, the parametric UMAP can handle highly nonlinear data better than UMAP. 
Note that the learnable parameters in UMAP are the embedding of points but the learnable parameters in the parametric UMAP are the weights of neural network so that the embedding is obtained by network.
An advantage of parametric UMAP over parametric t-SNE is that the probability of UMAP in the embedding space, Eq. (\ref{equation_UMAP_q}), does not have a normalization factor; hence, there is no need to normalize over the whole dataset. This makes UMAP easy to be used in deep learning.

The algorithm of parametric UMAP is shown in Algorithm \ref{algorithm_parametric_umap}. As this algorithm shows, we make mini-batches of data points where $\mathcal{B}_s$ denotes the $s$-th batch. In every epoch, we iterate over mini-batches and for every mini-batch, we iterate twice over the points of batch to have anchors and positive points. We optimize the cost function of UMAP over the batch and not the entire data. We denote the attractive and repulsive cost functions by $c_s^a$ and $c_s^r$, respectively. 
According to Eq. (\ref{equation_umap_cost2_withAttractiveAndRepulsiveForces}), the loss function of neural network in parametric UMAP should be $(c_s^a + c_s^r)$. Backpropagating this loss function trains the parameters $\theta$ of the neural network denoted by $f_\theta(.)$. After training, the embeddings are obtained as $\{f_\theta(\b{x}_i)\}_{i=1}^n$.
Note that as was discussed in Section \ref{section_supervised_semisupervised_umap}, one can combine the UMAP's cost function with cross-entropy loss or triplet loss to have supervised or semi-supervised embedding. Moreover, combining UMAP's cost function with reconstruction loss can train an autoencoder for embedding in its middle code layer.

\SetAlCapSkip{0.5em}
\IncMargin{0.8em}
\begin{algorithm2e}[!t]
\DontPrintSemicolon
    \textbf{Input}: New batch of training data $\b{X}_\text{new} = \{\b{x}_i\}_{i=1}^n$\;
    $\b{X}_\text{new}, \b{X}_\text{updated} \gets$ Update $k$NN graph by PANENE method\;
    \uIf{it is initial batch}{
        Initialize $\{\b{y}_i\}_{i=1}^n$ by Laplacian eigenmap\;
    }
    \Else{
        \For{each new point $\b{x}_i$ in $\b{X}_\text{new}$}{
            Find nearest neighbor to previously accumulated data\;
            Initialize $\b{y}_i$ to the embedding of the nearest neighbor point plus Gaussian noise\;
        }
    }
    \For{each new or updated point $\b{x}_i$ in $\b{X}_\text{new}$ or $\b{X}_\text{updated}$}{
        Calculate/update $\rho_i$ and $\sigma_i$ by Eqs. (\ref{equation_umap_rho}) and (\ref{equation_umap_sigma})\;
        Calculate $p_{ij}$ and $q_{ij}$ for $\forall j$ by Eqs. (\ref{equation_UMAP_p}) and (\ref{equation_UMAP_q})\;
    }
    \While{not converged}{
        \For{each new or updated point $\b{x}_i$ in $\b{X}_\text{new}$ or $\b{X}_\text{updated}$}{
            \For{$j$ from $1$ to $n$}{
                $u \sim U(0,1)$\;
                \If{$u \leq p_{ij}$}{
                    $\b{y}_i \gets \b{y}_i - \eta \frac{\partial c^a_{i,j}}{\partial \b{y}_i}$\;
                    \For{$m$ iterations}{
                        $l \sim U\{1, \dots, n\}$\;
                        $\b{y}_i \gets \b{y}_i - \eta \frac{\partial c^r_{i,l}}{\partial \b{y}_i}$\;
                    }
                }
            }
        }
    }
    \textbf{Return} embedding $\{\b{y}_i\}$ for the new and updated points\;
\caption{Progressive UMAP algorithm}\label{algorithm_umap_progressive}
\end{algorithm2e}
\DecMargin{0.8em}

\section{Progressive UMAP for Streaming and Out-of-sample Data}\label{section_progressive_UMAP}

The original UMAP does not support out-of-sample (test) data embedding. Progressive UMAP \cite{ko2020progressive} can embed out-of-sample data. It can also be used for embedding streaming (online) data. Although UMAP is generally faster than t-SNE, it takes some noticeable time to embed big data. Progressive UMAP can be used to embed some portion of data and then complete the embedding by embedding the rest of data as streaming data. 

The algorithm of progressive UMAP is shown in Algorithm \ref{algorithm_umap_progressive}. It uses PANENE \cite{jo2018panene} for constructing a streaming $k$NN graph. PANENE uses randomized kd-trees \cite{muja2009fast} for approximating and updating $k$NN graph. When a new batch of data, $\b{X}_\text{new}$, is arrived, some of data points are affected in $k$NN graph because their neighbors are changed or they become new neighbors of some points or they are no longer neighbors of some points. We denote the affected points by $\b{X}_\text{updated}$. If $\b{X}_\text{new}$ is the first batch of data, we embed data by Laplacian eigenmap, exactly as the original UMAP does. If the coming batch is not the first batch, we find the nearest neighbor of every new point among the previously accumulated data. Then, we set the initial embedding of every new point as the embedding of its nearest neighbor point added with some Gaussian noise. 
As in the original UMAP, for every new or updated point, we calculate or update $\rho_i$ and $\sigma_i$ by Eqs. (\ref{equation_umap_rho}) and (\ref{equation_umap_sigma}).
For every new point $\b{x}_i$, we also calculate $p_{ij}$ and $q_{ij}$ by Eqs. (\ref{equation_UMAP_p}) and (\ref{equation_UMAP_q}), where $j$ indexes all existing and new points. Then, for every new or updated point, we update the embedding of point by gradient descent where $m$ negative samples are used as in the original UMAP. The paper of progressive UMAP has not mentioned updating the embedding of neighbor (positive) points which we have in the original UMAP.

\section{Conclusion}\label{section_conclusion}

This was a tutorial paper on UMAP and its variants. We explained the details of UMAP and derived its gradients. We explained the theory of UMAP's cost function by algebraic topology and fuzzy category theory. We compared UMAP with t-SNE and LargeVis, discussed the repulsive forces and negative sampling, and introduced DensMAP, parametric UMAP, and progressive UMAP.  
For brevity, some other algorithms such as using genetic programming with UMAP \cite{schofield2021using} was not explained in this paper.




\bibliography{References}
\bibliographystyle{icml2016}

\end{document}